\documentclass{article}

\usepackage[margin=3.5cm, top=2.5cm, bottom=3cm]{geometry}
\usepackage[T1]{fontenc}
\usepackage{graphicx}

\usepackage{xcolor}
 
\usepackage{amsthm,amsmath,amsfonts,amssymb,dsfont}
\usepackage{etoolbox}
\usepackage{tikz}
\usetikzlibrary{calc}
\usepackage{marvosym}
\usepackage[colorlinks,
linkcolor=red!70!black,
citecolor=green!50!black,
urlcolor=magenta!60!black]{hyperref}
\usepackage[nameinlink]{cleveref}
\usepackage{authblk}
\usepackage{array}

\definecolor{rougeliens}{RGB}{168,55,44}
\definecolor{couleurHugo}{RGB}{1, 99, 183}
\definecolor{vertcitations}{RGB}{0,120,39}
\definecolor{bleuurl}{RGB}{59,92,184}
\definecolor{couleurmacro}{RGB}{121, 70, 167}
\hypersetup{
    colorlinks=true,
    linkcolor=rougeliens,
    citecolor=vertcitations,
    urlcolor=bleuurl
}

\newcommand{\commentaire}[3]{{\scriptsize\textbf{\color{#1}#2: \boldmath{#3}}}}
\newcommand{\Hugo}[1]{\commentaire{couleurHugo}{Hugo}{#1}}
\renewcommand{\Hugo}[1]{}

\newcommand{\macrostyle}[1]{#1}

\newcommand{\gne}{\macrostyle{\tilde{n}}}
\newcommand{\boucle}[1]{\tilde{#1}}
\newcommand{\retour}[1]{\overline{#1}}
\newcommand{\aretes}[1]{\mathrm{edges}(#1)}
\newcommand{\petales}[1]{\mathrm{petals}(#1)}

\newcommand{\mots}[1]{W(#1,\bbF)}
\newcommand{\code}[2]{\calC[#1,#2]}
\newcommand{\coupe}[2]{\mathrm{Cut}[#1,#2]}
\newcommand{\pliage}[2]{\mathrm{Fold}[#1,#2]}
\newcommand{\superposition}[1]{\macrostyle{\mathrm{mul}}(#1)}

\newcommand{\PSL}{\mathrm{PSL}}
\newcommand{\PGL}{\mathrm{PGL}}
\newcommand{\lambbar}{\bar{\lambda}}
\newcommand{\lambgne}{\tilde{\lambda}}

\newcommand{\rounds}{\macrostyle{R}}
\newcommand{\repetitions}{\macrostyle{L}}
\newcommand{\proportioncheck}{\macrostyle{\mu}}

\newcommand{\nombrecoupes}{\macrostyle{m}}
\newcommand{\irepetitions}{\macrostyle{\ell}}
\newcommand{\realirounds}{r}
\newcommand{\irounds}{\macrostyle{\realirounds}}\newcommand{\iroundsmax}{\macrostyle{\mathrm{\realirounds}}}
\newcommand{\iaretes}{\macrostyle{j}}
\newcommand{\sousensemblearetes}{\macrostyle{J}}
\newcommand{\icoupes}{\macrostyle{i}}
\newcommand{\randomness}{\macrostyle{\rho}}
\renewcommand{\path}{\macrostyle{p}}
\newcommand{\constant}{\macrostyle{\kappa}}

\newcommand{\rejectevent}{\macrostyle{B}}
\newcommand{\feven}{\hat{f}_{\text{even}}}
\newcommand{\fodd}{\hat{f}_{\text{odd}}}

\newcommand{\QED}{\qed}

\docsvlist{Cay,diam,RS}

\renewcommand{\leq}{\leqslant}
\renewcommand{\epsilon}{\varepsilon}
\docsvlist{A,B,C,D,E,F,G,H,I,J,K,L,M,N,O,P,Q,R,S,T,U,V,W,X,Y,Z,a,b,c,d,e,f,g,h,i,j,k,l,m,n,o,p,q,r,s,t,u,v,w,x,y,z}

\definecolor{charcoal}{HTML}{2E4052}
\tikzstyle{mininoeud}=[draw=black, thick, circle, text width=2mm, align=center, inner sep=0]
\tikzstyle{noeudselectionne}=[fill=cyan]
\tikzstyle{areteA}=[draw=charcoal, thick]
\tikzstyle{areteB}=[draw=charcoal, thick]
\tikzstyle{coupe}[cyan!60!blue]=[draw=#1, thick, densely dashed]
\tikzstyle{noeudscoupe}[cyan!60!blue]=[text=#1]
\tikzstyle{bouclage}=[looseness=14]
\tikzstyle{miniboucle}=[looseness=10]
\tikzstyle{ciseaux}=[]

\newtheorem{theorem}{Theorem}
\newtheorem{proposition}{Proposition}

\newtheorem{lemma}{Lemma}

\theoremstyle{definition}
\newtheorem{definition}{Definition}

\renewcommand{\QED}{}

\title{Codes on any Cayley Graph have an Interactive Oracle Proof of Proximity}
\author[1,2]{Hugo Delavenne}
\author[1,2]{Louise Lallemand}
\affil[1]{LIX, École Polytechnique, Institut Polytechnique de Paris}
\affil[2]{INRIA}
\date{}

\begin{document}

\maketitle%

\begin{abstract}
  Interactive Oracle Proofs of Proximity (IOPP) are at the heart of code-based SNARKs, a family of zeroknowledge protocols.
  The first and most famous one is the FRI protocol~\cite{BBHR18a}, that efficiently tests proximity to Reed-Solomon codes.
  This paper generalizes the flowering IOPP introduced in~\cite{DMR25} for some specific $(2,n)$-regular Tanner codes to a much broader variety of codes: any code with symbols indexed on the edges of a Cayley graph.
  The flowering protocol of~\cite{DMR25} had a soundness parameter much lower than the FRI protocol~\cite{BCIKS23}, and complexity parameters that could compete with the FRI~\cite{BBHR18a}.
  The lower soundness and the absence of restriction on the base field may lead to other practical speedups, however the codes considered in~\cite{DMR25} have an $o(1)$ minimum distance.
  The generalization proposed in this paper preserves the soundness parameter with a slight decrease of the complexity parameters, while allowing being applied on codes with constant rate and constant minimum distance thanks to the good expansion properties of some families of Cayley graphs.
\end{abstract}

\section{Introduction}
\subsection{Code-based SNARKs}

Succinct Non-interactive ARguments of Knowledge (SNARK) allow a computationally powerful entity, called the Prover, to convince a weaker entity, called the Verifier, that it has correctly performed a long computation.
The computation is stated for instance as a program and an input.
The Prover provides the output of the execution as well as a proof that the output was actually obtained from the given program on the given input.
The goal is for that proof to be much shorter than the length of the computation, and efficiently verifiable, that is verifying the proof should be much faster than running the program again.

The main practical application of SNARKs is for blockchains.
It allows to perform huge computations out of chain and only add the SNARK proof to the chain.
SNARKs can be applied on very general computations and other applications emerge.
It can preserve authenticity signature~\cite{DCB24}, for instance when one has a camera authenticating pictures as being from the real world and wants to crop or compress a picture to JPEG while preserving the authentication. The SNARK would be a proof of the correct use of the image editing program.
It can also be used to create AI regulation~\cite{FMZLXY24} by turning the training process into a SNARK. Other properties can then be proved on the training, for instance that the training data was free of use, or was unbiased.

One of the main approaches to build SNARKs relies on efficient proximity tests to error-correcting codes.
The program is first written in an arithmetic form in the \emph{arithmetization} phase.
It can be an arithmetic circuit for the R1CS~\cite{BCRSVW19} or PlonK~\cite{GWZC19} arithmetizations, or directly as a system of polynomial constraints for the AIR~\cite{BBHR18b} arithmetization.
Verifying the computation is then reduced to testing proximity of words built from the arithmetized computation to a Reed-Solomon code.
These words satisfy the property that if the computation is correct then they belong to a Reed-Solomon code, but if the computation is not correct then at least some of them are very far from the Reed-Solomon code.
This last gap obtained for non-correct computations explains why one can rely on proximity testing rather than code membership.

An important issue about SNARKs is that they often imply intrinsic arithmetic constraints on the field used.
For instance the existence of pairing friendly curves~\cite{AEG23} or a prime characteristic $p$ with $p-1$ being very smooth~\cite{BBHR18a,ACFY24}.
This hinders applications when the field is imposed and cannot be changed by the SNARK designer.
This is the case for example when trying to compose or interoperate SNARKs, or when dealing with signatures using the ``Bitcoin curve'' \texttt{secp256k1}~\cite{SEC2-2000} that require the field to be $\bbF_p$ with $p=2^{256}-2^{32}-977$.

\subsection{Interactive Oracle Proofs of Proximity}

Given a finite field $\bbF$ and $n\in\bbN$, a (linear) error-correcting code is a linear subspace of $\bbF^n$.
Two important parameters of a code are its dimension and its minimum distance for a given distance over $\bbF^n$, often the Hamming distance $\Delta_H(u,u'):=\frac1n|\{\iaretes\in[n]\mid u_\iaretes\neq u'_\iaretes\}|$.
The minimum distance of the code is the minimum distance between two distinct codewords.

Codes with efficient (non-interactive) proximity testing are a recent breakthrough~\cite{DELLM21} but those codes are not easily suitable for arithmetization.
The idea introduced in~\cite{BCS16} is to allow the proximity test to be made interactively between the Verifier and the Prover.
There can be several rounds of interaction.
At each round, the Verifier sends randomness to the Prover, and the Prover gives the Verifier an \emph{oracle access} to a word.
That means that the Prover commits to a word $u=(u_1,...,u_n)$, and the Verifier can later ask to get a coordinate $u_\iaretes$ for some $\iaretes$ and be assured that it gets $u_\iaretes$ and not $\tilde{u}_\iaretes\neq u_\iaretes$.
Such a proof is thus called an Interactive Oracle Proof of Proximity (IOPP).

Since the protocol involves randomness from the Verifier and that it does not read words entirely, the probability of failure is nonzero.
More precisely, there are two ways for the protocol to fail.
There could be false negatives, if the word tested $u$ is indeed a codeword of $C$ but that the Verifier rejects after the execution of the protocol. This property is called the \emph{completeness} and is the probability for the Verifier to accept a valid codeword. One can usually obtain a perfect completeness, that is a completeness of $1$.
There could also be false positives, if the word tested $u$ is very far from $C$ but that the Verifier accepts after the execution of the protocol. This property is called the \emph{soundness} and is the probability for the Verifier to accept a word $u$ too far from $C$. 
One tries to built IOPP with a soundness as small as possible.
It is important to notice that completeness gives the probability for the Verifier to accept if $u\in C$ whereas the soundness gives the probability for the Verifier to accept given that $\Delta_H(u,C)\geq\delta$ with $\delta$ a fixed gap parameter.

There are several complexity aspects have to be considered for such interactive protocols.
First, we count the number of rounds of the protocol, called the \emph{round complexity}.
Then we must consider the computational complexity of the Prover and the Verifier executing honestly the protocol, called the \emph{prover complexity} and the \emph{verifier complexity}.
We count the total length of commited messages to the Verifier as oracles, called the \emph{proof length}.
We also count the number of queries to the oracles that the Verifier asks, called the \emph{query complexity}, and the amount of randomness it is using, called the \emph{randomness complexity}.

When the Verifier can have public coin randomness, the IOPP is actually turned into a non-interactive proof using a Fiat-Shamir-like heuristic~\cite{FS87,BCS16}, and it is succinct thanks to the oracles instantiated with Merkle trees~\cite{Mer79}.
These two primitives rely on cryptographic hash functions, hence the only security assumption that we assume is that the Prover cannot find collisions for that hash function, and that the output of that hash function can be considered random.
Hence we only assume the random oracle model.

\subsection{Fast Reed-Solomon IOPP}

IOPPs allow proximity testing to arithmetization-friendly codes such as Reed-Solomon codes using the Fast Reed-Solomon IOPP (FRI) protocol~\cite{BBHR18a}.
Given $n$ distinct values $x_1,...,x_n\in\bbF$, the Reed-Solomon code of length $n$ and dimension $k\leq n$, denoted $\RS[n,k]$, is the vector space of evaluation of polynomials of degree $\leq k-1$ on $x_1,...,x_n$:
\[
  \RS[n,k]:=\{(\hat{f}(x_1),...,\hat{f}(x_n))\mid \hat{f}\in\bbF[X]_{\leq k-1}\}\subseteq\bbF^n\text.
\]
We can see Reed-Solomon codes as the words $f=(f_1,...,f_n)$ whose least degree interpolating polynomial $\hat{f}$ on $x_1,...,x_n$ has degree $\leq k-1$.
The idea of the FRI protocol is to reduce testing the proximity to a word $f$ of length $n$, to testing the proximity to a word of length $n/2$.
For $k$ even, we have that a polynomial $\hat{f}$ has degree $\leq k-1$ if, and only if, its even and odd parts, denoted $\feven$ and $\fodd$, have degree $\leq (k-2)/2$.
But proving successively that $\feven$ and $\fodd$ have the right degree would not gain in complexity compared to proving that $\hat{f}$ has the right degree.
Therefore the FRI protocol reduces testing the degree of $\hat{f}$ to testing the degree of $\feven+\randomness\fodd$ using the randomness $\randomness$ given by the Verifier. This linear combination is called the \emph{folding} of $\hat{f}$.
It has two important properties.
First, the folding at value $y=x^2$ can be computed using only two values of $f$ since we have that
\begin{align}
  \feven(y)&=\frac{\hat{f}(x)+\hat{f}(-x)}{2}
  &\fodd(y)&=\frac{\hat{f}(x)-\hat{f}(-x)}{2x}\text.\label{eq:folding-FRI}
\end{align}
Second, it can be can proven~\cite[Theorem 10]{BKS18} that if $f$ is far from $\RS[n,k]$, then the folding will be far from $\RS[n/2,k/2]$ with high probability over the randomness of the Verifier.
This property is called the \emph{commit soundness}.
Hence applying successively the folding operator on a word far from the original Reed-Solomon code will preserve a gap of distance to the successive codes.
However the issue of that folding is that it requires the field used to have $2^m$ roots of the unity with $2^m\approx n$ in order to be able to apply \cref{eq:folding-FRI} on successive domain of evaluation being the squares of the previous ones.

The FRI protocol on the input $f_0$ runs as follows.
During the commit phase, the Verifier sends randomness to the Prover which then gives oracle access to a word $f_1$, and they repeat for $\rounds$ rounds, the Prover giving oracle to $f_\irounds$ at each round $\irounds\in[\rounds]$.
Since the words commited by the Prover have a priori no link between each other, the Verifier checks during the query phase that each $f_\irounds$ is indeed built as the folding of $f_{\irounds-1}$ by testing on some random values. Finally, the Verifier checks that $f_\rounds$ has the right degree.
The FRI protocol has perfect completeness.
The soundness parameter has a component giving the probability that the commit phase goes wrong, the commit soundness, computed here in \Cref{proposition:commit-soundness}, and a component giving the probability that the query phase goes wrong, called the query soundness, computed here in \Cref{proposition:query-soundness}.

The STIR IOPP~\cite{ACFY24} allows the Verifier to query less coordinates. It works by folding and expanding the resulting domain, leading to a much faster decrease in the rate of the codes.
However, both FRI and STIR protocol require the field $\bbF$ to be very smooth: the FRI requires $|\bbF|-1$ to have a lot of small factors and the STIR requires $|\bbF|-1$ to be divisible by a large power of $2$.

\subsection{Codes on graphs}

We work here on the flowering protocol~\cite{DMR25}, which it an IOPP for some codes built on graphs, that mimics the FRI protocol by defining a different folding operator.
We consider codes on graphs similar to the ones used by Sipser and Spielman~\cite{SS96} where the graph is $n$-regular and the symbols of the words are values on the edges between the vertices.
Given a base code that will be here a Reed-Solomon code $\RS[n,k]$, the code is defined as the set of words such that the ``local view'' around each vertex, composed of the vector of values of edges around that vertex, is a Reed-Solomon codeword.
These codes can also be seen as Tanner codes.
A Tanner graph is a bipartite graph with the symbols on one side and the constraints on the other side, and the edges link symbols to constraints to be satisfied.
With our graphs, the symbols (the values on the edges) would be linked to two constraints (the vertices that require their local view to be a Reed-Solomon codeword).

The graphs we consider are Cayley graphs~\cite{Cay78}.
Given a group $G$ and a symmetric generating set $S\subseteq G$, the Cayley graph $\Cay(G,S)$ has vertices the elements of $G$ and there is an edge between $g$ and $g'\in G$ if there exists $s\in S$ such that $g'=g\cdot s$.
We are interested in Cayley graphs for the symmetries they provide and that helps to create a folding operator, but also for the good \emph{expansion} of some of these graphs~\cite{Mur20,Lub12}.
These expansion properties imply both a lower bound on the minimum distance of graph codes built on them \cite{K19}, and an upper bound on their diameter~\cite[Lemma 2.1]{DELLM21}.

\subsection{Flowering protocol}

An IOPP, called \emph{Flowering}, for codes on graphs was proposed in~\cite{DMR25}.
The authors create a folding operator to mimic the FRI protocol.
They do not split the codeword in two halves, instead they take a coarser view and split the set of vertices of the graph in two halves, and create two words with the edges of the two resulting subgraphs, that will play the role of the even and odd parts of the FRI.
This operation is called \emph{cutting} the graph.
When the two cut graphs are isomorphic, it is possible to define a folding operator by simply doing a linear combination of the codewords using that isomorphism.
Then the protocol is the same as the FRI protocol.

Since the number of vertices is divided by two at each step, the graph must have a number of edges that is a power of two.
To achieve this and having a lot of symmetries to create isomorphic cuts, the authors apply their protocol to Cayley graphs of the additive group over $(\bbZ/2\bbZ)^r$.
Using this, they achieve complexity parameters close to the FRI protocol, and soundness parameters that surpass the FRI protocol.
However the codes built on these graphs have a relative minimal distance that goes to $0$ as the length of the code increases.

\subsection{Results}

We generalize the Flowering protocol from~\cite{DMR25} by allowing to separate the vertices in more than two sets, and most importantly to sets of vertices that may overlap and thus not form a partition of the vertices.
While this reduces less the size of the folding compared to exactly dividing the size, and yields higher factors for the Prover complexity, it allows us to construct a folding operator for Cayley graphs on any group and any symmetric generating set.
Using the expansion properties of Cayley graphs, we can then achieve an IOPP for codes on graphs with a constant rate and constant relative minimum distance.

We manage to keep the same soundness parameters as in~\cite[Theorem 2]{DMR25} by giving different weights to the vertices depending on the number of cuts they are part of.
Therefore we achieve an IOPP with constant rate and minimum distance, and lower soundness than the FRI and STIR protocols, since our protocol require a smaller field and can be applied up to the covering radius whereas the FRI and the STIR protocols are only proven up to the Johnson radius~\cite[Theorem 7.2]{BCIKS23}, while being conjectured up to the covering radius~\cite[Conjecture 2.3]{BGKS20}.

We compare in \Cref{tab:comparison-complexity,tab:comparison-soundness} the complexity and soundness parameters of the FRI protocol~\cite{BBHR18a,BCIKS23} and the STIR protocol~\cite{ACFY24} with the original flowering protocol~\cite{DMR25} and our generalization, stated more precisely in \Cref{theorem:complexities,theorem:soundness}.
All protocols are applied on a code of length $N$ and dimension $K$.
The codes used for the flowering protocol are those built in \Cref{sec:application}, where $n$ and $\constant$ are constants.
Note that the $\log N$ factors of the prover complexity and proof length are due to very rough estimations of the sizes of intermediate codes due to the overlap of the set of vertices, since they are considered all constant for the estimation.

  \begin{table*}[th!]
    \centering
      \begin{tabular}{m{15mm}|m{28mm}|c|c|c}
        \centering{}Protocol & \centering{}FRI \cite{BBHR18a,BCIKS23} & STIR~\cite{ACFY24} & Flowering~\cite{DMR25} & Flowering (here) \\[2pt]\hline&&&&\\[-8pt]
        \centering{}Prover & \centering$<8N$ & $O(N)$ & $<3N$ & $<5\constant N\log N$ \\[2pt]&&&&\\[-8pt]
        \centering{}Verifier & \centering$<\frac{2\lambda\log K}{\min(\delta,1-\sqrt{K/N})}$ & $O(\lambda^2+\lambda\log\log K)$ & $<\frac{4\lambda\log^2N}{\delta+\log N/N}$ & $<\frac{8\lambda n\constant\log N}{\delta+\log N/N}$ \\[2pt]&&&&\\[-8pt]
        \centering{}Query & \centering$\frac{2\lambda\log K}{\min(\delta,1-\sqrt{K/N})}$ & $O(\lambda\log\log K)$ & $<\frac{2\lambda\log^2N}{\delta+\log N/N}$& $<\frac{3\lambda n\constant\log N}{\delta+\log N/N}$ \\[2pt]&&&&\\[-8pt]
        \centering{}Rounds & \centering$\log K$ & $O(\log K)$ & $<\log N$ & $<\constant\log N$\\[2pt]&&&&\\[-8pt]
        \centering{}Length & \centering$<N$ & $N + O(\log K)$ & $<N$ & $<\constant N\log N$ \\[2pt]\hline&&&&\\[-8pt]
      \centering{}Field size & \centering$> \frac{2^\lambda 10^7N^{3.5}\log K}{K^{1.5}}$ & $\Omega\left(\frac{\lambda 2^\lambda K^2 N^{3.5}}{\log (N/K)}\right)$ & $>2^\lambda N\log N$&$>2^{\lambda+1}\constant N\log N$\\[2pt]&&\\[-8pt]
      \centering{}Field structure & \centering$|\bbF|-1$ has lots of small factors & $\left.2^{\lceil\log K\rceil}\,\right|\, |\bbF|-1$ & \multicolumn{2}{c}{Any}
    \end{tabular}
  \caption{Comparison of the complexity parameters to reach $\lambda$ bits of security when testing $\delta$-proximity to a code of length $N$ and dimension $K$.
    The constants for the STIR depend on $\delta$. Here, $\constant$ and $n$ are constant, $n:=p+1$ and $\constant:=\constant(p)$ with the notations of \Cref{sec:case-study}.}
  \label{tab:comparison-complexity}
\end{table*}

\begin{table*}[th!]
  \centering
    \begin{tabular}{m{15mm}|c|c|c}
      \centering{}Protocol & FRI \cite{BBHR18a,BCIKS23} & Flowering~\cite{DMR25} & Flowering (here) \\[2pt]\hline&&&\\[-8pt]
      \centering{}Commit soundness & $\displaystyle\frac{10^7N^{3.5}\log K}{K^{1.5}|\bbF|}$ & $\displaystyle\frac{N\log N}{|\bbF|}$ & $\displaystyle\frac{2\constant N\log N}{|\bbF|}$\\[2pt]&\\[-8pt]
      \centering{}Query soundness & $\displaystyle\left(1-\min\left(\delta,1-1.05\sqrt{\frac{K}{N}}\right)\right)^\repetitions$ & \multicolumn{2}{c}{$\displaystyle\left(1-\delta-\frac{\log N}{N}\right)^\repetitions$}
    \end{tabular}
  \caption{Comparison of the soundness parameter when testing $\delta$-proximity to a code of length $N$ and dimension $K$, where $\repetitions\in\bbN$ is the repetition parameter, and $\constant$ and $n$ are constant, $n:=p+1$ and $\constant:=\constant(p)$ with the notations of \Cref{sec:case-study}. The total soundness is the sum of the commit and query soundness.}
  \label{tab:comparison-soundness}
\end{table*}

The typical range of length and dimension when using the FRI or the STIR, as stated in the benchmark of~\cite{ACFY24} is $K\in[2^{18},2^{30}]$ and $N\in[2K,16K]$.
For $N=2^{19}$, $K=2^{18}$ and $\lambda=128$ for instance, the STIR requires a field of size $>2^{237}$ and the FRI a field of size $>2^{194}$, while our protocol requires a generic field of size $>2^{160}$ with $\constant\approx 2^7$.
This means that our protocol could benefit in term of Prover and Verifier complexity from working on smaller fields.



\section{Graphs and codes on graphs}

In the paper, rounds of protocols will be indexed using the letter $\irounds\in[\rounds]$, edges outgoing from a vertex will be indexed using the letter $\iaretes\in[n]$, cuts will be indexed using the letter $\icoupes\in[\nombrecoupes]$, and steps of repetitions of some phases using the letter $\irepetitions\in[\repetitions]$.

\subsection{Regular Indexed Multigraphs}

Fix a finite field $\bbF$ and $n\leq|\bbF|$.
We consider $n$-regular indexed multigraphs ($n$-RIM), which are graphs with multiple edges between vertices and loops around vertices, and whose edges around a vertex are indexed between $1$ and $n$.
We slightly extend the notion from~\cite{DMR25} by allowing an edge between $v$ and $v'$ to be indexed differently in $v$ and in $v'$.

\begin{definition}[Regular Indexed Multigraph (RIM) {\cite[Definition 1]{DMR25}}]
  A \emph{$n$-RIM} $\Gamma=(V,E)$ is given by a set of vertices $V$ and a function $E:V\times[n]\to V$ such that for $(v,\iaretes)\in V\times[n]$, $E(v,\iaretes)$ is the neighbor of $v$ through the edge indexed $\iaretes$, and such that $E$ satisfies the \emph{well-definedness property}: there is a permutation of $[n]$ denoted $\iaretes\mapsto\retour{\iaretes}$ such that\Hugo{il faudrait qu'on définisse ça dans nos graphes partie 4}
  \begin{equation*}
    \label{eq:well-definedness}
    \forall (v,\iaretes)\in V\times[n], \hskip 3mm E\!\left(E(v,\iaretes),\retour{\iaretes}\right)=v\text.
  \end{equation*}
\end{definition}

\noindent
We will consider that the permutation $\iaretes\mapsto\retour{\iaretes}$ is always fixed and clear from context.
Given an $n$-RIM $\Gamma=(V,E)$, its set of edges, denoted $\aretes{\Gamma}$, is defined as the quotient of $V\times[n]$ by the relation $\sim_E$ defined by $(v,\iaretes)\sim_E(v',\iaretes')$ iff either $(v,\iaretes)=(v',\iaretes')$, or $\iaretes=\retour{\iaretes'}$ and $E(v,\iaretes)=v'$:
\begin{equation*}
  \label{eq:edges}
  \aretes{\Gamma}=(V\times[n])/\sim_E\text.
\end{equation*}

\noindent
Note that the equivalence classes of $\sim_E$ have either one element for loop edges, or two elements for edges between distinct vertices.
We denote by $\petales{\Gamma}\subseteq\aretes{\Gamma}$ the set of edges $(v,\iaretes)$ that are loops, and that we call \emph{petals}.

We study words on graphs whose symbols are indexed on the edges of the graphs. We see these words as functions $\aretes{\Gamma}\to\bbF$ and we denote by $\mots{\Gamma}$ the set of words on a graph $\Gamma$ in $\bbF$.
For $f\in\mots{\Gamma}$ and $v\in V$, denote by $f(v,\cdot)=(f(v,1),...,f(v,n))$ the vector of the values indexed by the edges around $v$.
For clarity, we will actually use the words $f\in\mots{\Gamma}$ as functions $V\times[n]\to\bbF$, without proving explicitely that whenever $(v,\iaretes)\sim_E(v',\iaretes')$, then $f(v,\iaretes)=f(v',\iaretes')$.

Tanner codes~\cite{Tan81} are codes built on $(l,r)$-regular bipartite graphs where the symbols are vertices of arity $l$ on the left side and the constraints are vertices of arity $r$ on the right side.
The constraints are applied to the symbols connected to them.
Since the symbols of our words are on the edges, hence connected to two constraints, we can see our codes as Tanner codes on $(2,n)$-regular bipartite graphs by adding vertices on the edges to put symbols on them.
In the case studied in this paper, the constraints will be that the symbols connected to a vertex have to form a Reed-Solomon codeword.

\begin{definition}[Code on graph]
  Given an $n$-RIM $\Gamma=(V,E)$ and $k\leq n$, we define the code $\code{\Gamma}{k}\subseteq\mots{\Gamma}$ by
  \[
    \code{\Gamma}{k}:=\left\{f\in\mots{\Gamma}\mid \forall v\in V, f(v,\cdot)\in \RS[n,k]\right\}\text.
  \]
\end{definition}

\noindent
We give a general lower bound on the dimension of these codes.

\begin{proposition}[Lower bound on the dimension {\cite[Proposition 2]{DMR25}}]
  \label{proposition:bound-on-the-dimension}
  Let $\Gamma=(V,E)$ be an $n$-RIM and $k\leq n$.
  Then
  \begin{equation*}
    \dim\code{\Gamma}{k}\geq(k-n/2)|V|+|\petales{\Gamma}|/2\text.
  \end{equation*}
\end{proposition}

\noindent
We use the \emph{relative Hamming distance}, denoted $\Delta_H$ between words of $\mots{\Gamma}$, and the relative vertex distance from \cite{DMR25}.

\begin{definition}[Vertex distance {\cite[Definition 8]{DMR25}}]
  \label{definition:vertex-distance}
  Let $\Gamma=(V,E)$ be an $n$-RIM.
  Let $f,f'\in\mots{\Gamma}$.
  Define the relative \emph{vertex distance} between $f$ and $f'$, denoted $\Delta_V$, as
  \begin{equation*}
    \label{eq:vertex-distance}
    \Delta_V(f,f'):=\frac{1}{|V|}|\{v\in V\mid f(v,\cdot)\neq f'(v,\cdot)\}|\text.
  \end{equation*}
\end{definition}

\noindent
The vertex distance can be thought as a coarser version of the Hamming distance.

\begin{proposition}[Comparison between distances {\cite[Proposition 1]{DMR25}}]
  \label{proposition:inequality-distances}
  Let $\Gamma$ be an $n$-RIM.
  Let $f,f'\in\mots{\Gamma}$.
  Then there exists $\mu(\Gamma)\in[0,1]$ such that $\Delta_V(f,f')\geq\mu(\Gamma)\Delta_H(f,f')$.
  Furthermore, if all vertices of $\Gamma$ have the same number of petals then $\mu(\Gamma)=1$.
\end{proposition}

\noindent
In particular, if $\Gamma$ has no petals then $\mu(\Gamma)=1$.
Moreover, the generalization in this paper requires weighted vertex distances.

\begin{definition}[Weighted vertex distance]
  Let $w:V\to\bbR^+$ be a weight function and $T\subseteq V$.
  Denote
  \[
    |T|_w:=\sum_{v\in T}w(v)\text.
  \]
  If $f,f'\in\mots{\Gamma}$, denote $D_w(f,f'):=|\{v\in V\mid f(v,\cdot)\neq f'(v,\cdot)\}|_w$, and define the \emph{relative $w$-weighted vertex distance} by
  \[
    \Delta_w(f,f'):=\frac{D_w(f,f')}{|V|_w}\in[0,1]\text.
  \]
\end{definition}

\subsection{Cutting graphs}

In order to mimic the FRI protocol \cite{BBHR18a}, we need to find a way to split the words of $\mots{\Gamma}$ into several pieces.
To achieve this, we split the graph into several pieces according to some subsets of vertices, and cut the edges leaving these subsets.

\begin{definition}[Cut-graph, cut-word {\cite[Definition 4]{DMR25}}]
  If $\Gamma=(V,E)$ is an $n$-RIM, for $V'\subseteq V$, define the \emph{cut-graph} $\coupe{\Gamma}{V'}$ as the $n$-RIM $(V',E')$ where
  \begin{equation*}
    \label{eq:cut-graph}
    E':(v,\iaretes)\mapsto
    \begin{cases}
      E(v,\iaretes)&\text{if }E(v,\iaretes)\in V'\\
      v&\text{otherwise.}
    \end{cases}
  \end{equation*}
  If $f\in\mots{\Gamma}$ and $V'\subseteq V$, define the \emph{cut-word} $\coupe{f}{V'}\in\mots{\coupe{\Gamma}{V'}}$ as the restriction of $f$ to $\coupe{\Gamma}{V'}$: for $(v,\iaretes)\in V'\times[n], \coupe{f}{V'}(v,\iaretes)=f(v,\iaretes)$.
\end{definition}

\noindent
\Cref{fig:cut-graph} illustrates a graph cutting.
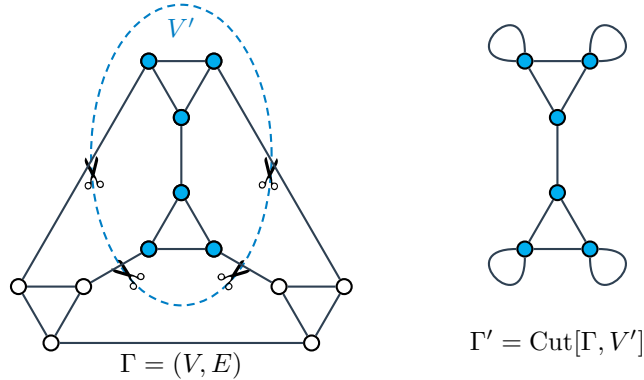
\begin{figure}[h!]
  \centering
  \begin{tikzpicture}
    \draw [coupe] (0,1) circle (1.2 and 2);
    \node [noeudscoupe] at (0,2.7) {$V'$};

    \foreach \a in {-7, -53} {
      \node [ciseaux, sloped, rotate=90+\a] at ($(0,1) + (\a:1.2 and 2)$) {\LARGE\LeftScissors};
      \node [ciseaux, sloped, rotate=180+90-\a] at ($(0,1) + (180-\a:1.2 and 2)$) {\LARGE\RightScissors};
    }

    \foreach \i in {0,1,2} {
      \node [mininoeud] (\i) at (\i/3*360+90:0.5) {};
      \foreach \j in {0,1,2} {
        \node [mininoeud] (\j\i) at ($(\i/3*360+90:2) + (\j/3*360+\i/3*360-90:0.5)$) {};
      }
    }
    \draw [areteB] (0) to (1) to (2) to (0);
    \foreach \i/\j in {0/2,1/0,2/1} {
      \draw [areteA] (0\i) to (\i);
      \draw [areteA] (1\i) to (2\j);
      \draw [areteB] (0\i) to (1\i) to (2\i) to (0\i);
    }
    \node [below] at (0,-1.5) {$\Gamma=(V,E)$};
    \foreach \i in {0,1,2} {
      \node [mininoeud, noeudselectionne] at (\i/3*360+90:0.5) {};
      \node [mininoeud, noeudselectionne] at ($(90:2) + (\i/3*360-90:0.5)$) {};
    }


    \begin{scope}[xshift=5cm]
      \foreach \i in {0,1,2} {
        \node [mininoeud, noeudselectionne] (A\i) at (\i/3*360+90:0.5) {};
        \node [mininoeud, noeudselectionne] (B\i) at ($(90:2) + (\i/3*360-90:0.5)$) {};
      }
      \draw [areteB] (A0) to (A1) to (A2) to (A0);
      \draw [areteB] (B0) to (B1) to (B2) to (B0);
      \draw [areteA] (A0) to (B0);
      \foreach \n [count=\i] in {B1,B2,A1,A2} {
        \draw [areteA, bouclage, in=90*\i-45+45, out=90*\i-45-45] (\n) to (\n);
      }
      \node [below] at (0,-1.2) {$\Gamma'=\coupe{\Gamma}{V'}$};
    \end{scope}
  \end{tikzpicture}
  \caption{Example of graph cutting. The cyan edges $V'\subseteq V$ are the ones kept for the cut. The edges leaving the subgraph are cut and become petals.}
  \label{fig:cut-graph}
\end{figure}
Once we have several cuts $V_0,...,V_{m-1}\subseteq V$, we also need a way to \emph{fold} their cut-words together as a single word, like for the FRI protocol.
We require to apply such a fold that the cuts are all \emph{isomorphic}.

\begin{definition}[RIM isomorphism {\cite[Definition 3]{DMR25}}]
  Let $\Gamma=(V,E)$ and $\Gamma'=(V',E')$ be $n$-RIM. An \emph{isomorphism} from $\Gamma$ to $\Gamma'$ is a bijection $\varphi:V\to V'$ such that for any $(v,\iaretes)\in V\times[n]$ we have $\varphi(E(v,\iaretes))=E'(\varphi(v),\iaretes)$.
\end{definition}

\begin{definition}[Flowering cut collection {\cite[Definition 5]{DMR25}}]
  For $\nombrecoupes\geq 2$, let $V_0,...,V_{\nombrecoupes-1}\subseteq V$ such that $\bigcup_{\icoupes=0}^{\nombrecoupes-1} V_\icoupes=V$.
  If for all $\icoupes\in\{0,...,\nombrecoupes-1\}$, $\coupe{\Gamma}{V_\icoupes}$ is isomorphic to $\coupe{\Gamma}{V_0}$, then we say that $(V_0,...,V_{\nombrecoupes-1})$ is a \emph{flowering cut collection} of order $\nombrecoupes$ from $\Gamma$ to $\coupe{\Gamma}{V_0}$.

  For $v\in V$, denote $\superposition{v}:=|\{\icoupes\in\{0,...,\nombrecoupes-1\}\mid v\in V_\icoupes\}|\geq 1$ the number of cuts that overlap on $v$, that we call \emph{multiplicity}.
  For the symmetries of the cuts, we furthermore require that a cut collection $(V_0,...,V_{\nombrecoupes-1})$ satisfies that the multiplicity is invariant under the isomorphisms, i.e. for any $v\in V$ and $\icoupes\in\{0,...,\nombrecoupes-1\}$,
  \begin{equation}
    \label{eq:symmetry-cuts}
    \superposition{v}=\superposition{\varphi_\icoupes(v)}\text.
  \end{equation}
\end{definition}

\noindent
\Cref{fig:cut-collection} illustrates a flowering cut collection.
The difference with \cite{DMR25} lies here: the flowering cuts collections were restricted to $\nombrecoupes=2$, and $V_0$ and $V_1$ being disjoint. 
\begin{figure}[h!]
  \centering
  \begin{tikzpicture}
    \foreach \i/\c in {0/cyan,1/magenta,2/yellow} {
      \node [mininoeud, fill=black] (\i) at (\i/3*360+90:0.5) {};
      \foreach \j in {0,1,2} {
        \node [mininoeud, fill=\c] (\j\i) at ($(\i/3*360+90:2) + (\j/3*360+\i/3*360-90:0.5)$) {};
      }
    }
    \draw [areteB] (0) to (1) to (2) to (0);
    \foreach \i/\j in {0/2,1/0,2/1} {
      \draw [areteA] (0\i) to (\i);
      \draw [areteA] (1\i) to (2\j);
      \draw [areteB] (0\i) to (1\i) to (2\i) to (0\i);
    }
    \node [below] at (0,-1.8) {$\Gamma=(V,E)$};
      \foreach \r/\c in {0/cyan!60!blue,1/magenta,2/red!40!yellow} {
      \begin{scope}[rotate=\r/3*360]
        \draw [coupe=\c] (0,1) circle (1.2 and 2);
        \node [noeudscoupe=\c] at (0,2.6) {$V_\r$};
      \end{scope}
    }

    \begin{scope}[xshift=6cm]
      \foreach \r/\c in {0/cyan,1/magenta,2/yellow} {
        \begin{scope}[rotate=\r/3*360, xscale=0.5, yscale=0.5]
          \coordinate (centre) at (90:2);
          \foreach \i in {0,1,2} {
            \node [mininoeud, fill=black] (A\i) at ($(centre) + (\i/3*360+90:0.5)$) {};
            \node [mininoeud, fill=\c] (B\i) at ($(centre) + (90:2) + (\i/3*360-90:0.5)$) {};
          }
          \draw [areteB] (A0) to (A1) to (A2) to (A0);
          \draw [areteB] (B0) to (B1) to (B2) to (B0);
          \draw [areteA] (A0) to (B0);
          \foreach \n [count=\i] in {B1,B2,A1,A2} {
            \draw [areteA, miniboucle, in=90*\i-45+45, out=90*\i-45-45] (\n) to (\n);
          }
        \end{scope}
      }
      \node at (0,2.8) {$\coupe{\Gamma}{V_0}$};
      \node at (-1.3,-2) {$\coupe{\Gamma}{V_1}$};
      \node at (1.3,-2) {$\coupe{\Gamma}{V_2}$};
    \end{scope}
  \end{tikzpicture}
  \caption{
    Example of flowering cut collection.
    All cut-graphs are isomorphic.
    The black vertices have a multiplicity of 3 since they are contained in all cuts, while the other cyan, magenta and yellow vertices only have a multiplicity of 1.
  }
  \label{fig:cut-collection}
\end{figure}
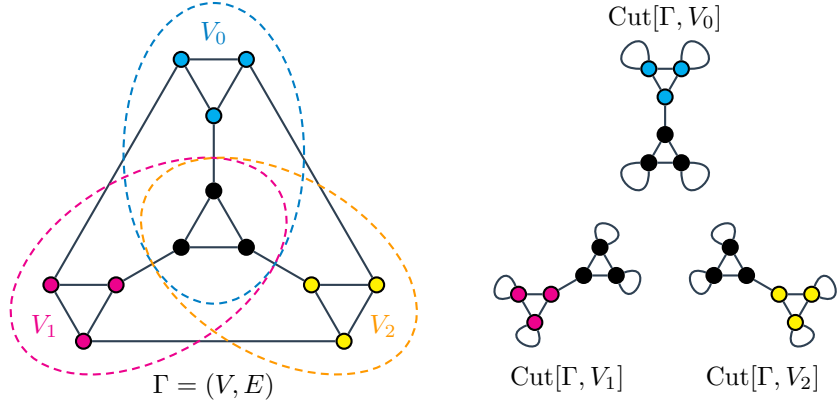

\begin{definition}[Folding {\cite[Definition 6]{DMR25}}]
  Let $\Gamma=(V,E)$ be an $n$-RIM and let $(V_0,...,V_{\nombrecoupes-1})$ be a flowering cut of $V$, using $\varphi_\icoupes:V_0\to V_\icoupes$ as RIM isomorphism for $\icoupes\in\{0,...,\nombrecoupes-1\}$.
  For $f\in\mots{\Gamma}$ and $\randomness\in\bbF$, we define the \emph{folding} of $f$ over $\randomness$, denoted $\pliage{f}{\randomness}$, as the word of $\mots{\coupe{\Gamma}{V_0}}$ such that
  \begin{equation*}
    \label{eq:fold}
    \pliage{f}{\randomness}:(v,\iaretes)\mapsto\sum_{\icoupes=0}^{\nombrecoupes-1}\randomness^\icoupes\coupe{f}{V_\icoupes}(\varphi_\icoupes(v),\iaretes)\text.
  \end{equation*}
\end{definition}

\noindent
We will want to apply recursively such cuts and foldings, until we reach a graph composed of a single vertex with only petals, that we call the \emph{flower}.

\begin{definition}[Blossoming graph sequence]
  Let $\Gamma_0,...,\Gamma_\rounds$ be a sequence of graphs such that $\Gamma_\rounds$ only has one vertex.
  If for every $\irounds\in[\rounds]$ there exists a flowering cut collections $(V_{\irounds,0},...,V_{\irounds,\nombrecoupes-1})$ from $\Gamma_{\irounds-1}$ to $\Gamma_\irounds$, then we say that $(\Gamma_0,...,\Gamma_\rounds)$ is a \emph{blossoming graph sequence} of $\Gamma_0$.
\end{definition}

\noindent
The blossoming graph sequences can also be defined for flowering cut collections of different order $\nombrecoupes_\irounds$ at each step $\irounds$.
\Cref{fig:blossoming-sequence} illustrates such a blossoming graph sequence, alternating cuts of order 3 and cuts of order 2.

\begin{figure}[h!]
  \centering
  \begin{tikzpicture}
    \foreach \i/\c in {0/magenta,1/cyan,2/yellow} {
      \node [mininoeud] (\i) at (\i/3*360+90:0.4) {};
      \foreach \j in {0,1,2} {
        \node [mininoeud] (\j\i) at ($(\i/3*360+90:1.5) + (\j/3*360+\i/3*360-90:0.4)$) {};
      }
    }
    \draw [areteB] (0) to (1) to (2) to (0);
    \foreach \i/\j in {0/2,1/0,2/1} {
      \draw [areteA] (0\i) to (\i);
      \draw [areteA] (1\i) to (2\j);
      \draw [areteB] (0\i) to (1\i) to (2\i) to (0\i);
    }
    \node at (0,-2) {$\Gamma_0$};
    \foreach \r in {0,1,2} {
      \begin{scope}[rotate=\r/3*360]
        \draw [coupe] (0,0.75) circle (1 and 1.5);
      \end{scope}
    }
    \node [noeudscoupe] at (0,2.6) {$V_{1,0}$};
    \node [noeudscoupe] at (-1.1,-1.8) {$V_{1,1}$};
    \node [noeudscoupe] at (1.1,-1.8) {$V_{1,2}$};

    \begin{scope}[xshift=3cm, yshift=-0.5cm]
      \foreach \i in {0,1,2} {
        \node [mininoeud] (A\i) at ($(\i/3*360+90:0.4)$) {};
        \node [mininoeud] (B\i) at ($(90:1.5) + (\i/3*360-90:0.4)$) {};
      }
      \draw [areteB] (A0) to (A1) to (A2) to (A0);
      \draw [areteB] (B0) to (B1) to (B2) to (B0);
      \draw [areteA] (A0) to (B0);
      \foreach \n [count=\i] in {B1,B2,A1,A2} {
        \draw [areteA, miniboucle, in=90*\i-45+45, out=90*\i-45-45] (\n) to (\n);
      }
      \draw [coupe] (0,-0.2) circle (0.8 and 0.85);
      \draw [coupe] (0,1.7) circle (0.8 and 0.85);
      \node [noeudscoupe] at (0,2.2) {$V_{2,1}$};
      \node [noeudscoupe] at (0,-0.7) {$V_{2,0}$};

      \node at (0,-1.5) {\scalebox{0.9}{$\Gamma_1:=\coupe{\Gamma_0}{V_{1,0}}$}};
    \end{scope}

    \begin{scope}[xshift=5.8cm]
      \foreach \i in {0,1,2} {
        \node [mininoeud] (\i) at ($(\i/3*360+90:0.5)$) {};
        \draw [areteA, miniboucle, in=360/3*\i+90-45, out=360/3*\i-45+180] (\i) to (\i);
        \draw [coupe] (\i/3*360+90:0.6) circle (0.45);
      }
      \node [noeudscoupe] at (0,1.3) {$V_{3,0}$};
      \node [noeudscoupe] at (-0.5,-1) {$V_{3,1}$};
      \node [noeudscoupe] at (0.5,-1) {$V_{3,2}$};
      \draw [areteB] (0) to (1) to (2) to (0);

      \node at (0,-2) {\scalebox{0.9}{$\Gamma_2:=\coupe{\Gamma_1}{V_{2,0}}$}};
    \end{scope}

    \begin{scope}[xshift=8.6cm]
      \node [mininoeud] (fleur) {};
      \foreach \i/\a in {0/areteA,1/areteB,2/areteB} {
        \draw [\a, miniboucle, in=360/3*\i+90-45, out=360/3*\i+90+45] (fleur) to (fleur);
      }
      \node at (0,-2) {\scalebox{0.9}{$\Gamma_3:=\coupe{\Gamma_2}{V_{3,0}}$}};
    \end{scope}
  \end{tikzpicture}
  \caption{Example of blossoming graph sequence. $\Gamma_3$ is a flower with only petals.}
  \label{fig:blossoming-sequence}
\end{figure}
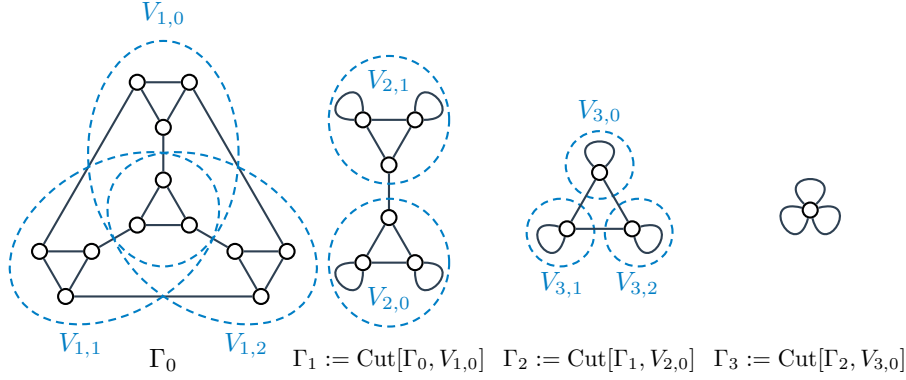

\noindent
Along with the blossoming graph sequence, we define recursively a \emph{sequence of weights} $w_0,...,w_\rounds$ such that, for $\irounds\in[\rounds]$, $w_\irounds:V_\irounds\to\bbR^+$.
Start by the uniform weight, $w_0:v\mapsto 1$. Therefore $\Delta_{w_0}=\Delta_V$.
Then $w_{\irounds+1}$ is defined on $V_{\irounds+1}$ as
\[
  w_{\irounds+1}(v):=\frac{w_\irounds(v)}{\superposition{v}}\text.
\]
Using \cref{eq:symmetry-cuts} we have by induction that these weights are invariant under the RIM isomorphisms, as stated in \Cref{proposition:symmetry-weight}.

\begin{proposition}[Symmetry of the weights]
  \label{proposition:symmetry-weight}
  Let $\Gamma_0,...,\Gamma_\rounds$ be a blossoming graph sequence and $w_0,...,w_\rounds$ be a sequence of weights defined as above.
  Then for $\irounds\in\{0,...,\rounds\}$, $v\in V_\irounds$ and $\icoupes\in\{0,...,\nombrecoupes-1\}$, we have
  \begin{equation}
    \label{eq:symmetry-weight}
    w_{\irounds}(v)=w_{\irounds}(\varphi_\icoupes(v))\text.
  \end{equation}
\end{proposition}


\section{Flowering protocol}

\subsection{Distance reduction probability}

By applying the folding operator, the Prover reduces testing proximity to $f\in\mots{\Gamma}$ to testing proximity to $f'\in\mots{\Gamma'}$, and one must ensure that this reduction does not make $f'$ too much closer to $\code{\Gamma'}{k}$ than $f$ was to $\code{\Gamma}{k}$.
The commit soundness from \cite[Proposition 3]{DMR25} do not apply exactly with the generalization of the cuts to $\nombrecoupes$ non-disjoint cuts.
We generalize this result in \Cref{proposition:commit-soundness} using the weighted distances.

\begin{proposition}[Commit soundness]
  \label{proposition:commit-soundness}
  Let $\epsilon>0$ and $k<n$.
  Let $\Gamma=(V,E)$ be an $n$-RIM and $V_0,...,V_{\nombrecoupes-1}\subseteq V$ be a flowering cut from $\Gamma$ to $\Gamma'=\coupe{\Gamma}{V_0}$, with $\varphi_\icoupes:V_0\to V_\icoupes$ as isomorphism.
  Let $w:V\to\bbR^+$ be a weight function invariant under the isomorphisms $\varphi_\icoupes$, i.e. satisfying \cref{eq:symmetry-weight}.
  Denote $\tilde{w}:V\to\bbR^+$ the weight function such that $\tilde{w}(v)=\frac{w(v)}{\superposition{v}}$.
  Denote $C:=\code{\Gamma}{k}$ and $C':=\code{\Gamma'}{k}$.
  Let $f\in\mots{\Gamma}$.
  Then
  \begin{equation*}
    \label{eq:commit-soundness}
    \underset{\randomness\in\bbF}{\Pr}\left[\Delta_{\tilde{w}}(\pliage{f}{\randomness},C')<\Delta_w(f,C)-\epsilon\right]\leq\frac{\nombrecoupes-1}{\epsilon|\bbF|}\text.
  \end{equation*}
\end{proposition}

\begin{proof}
  Denote $\delta:=\Delta_w(f,C)$ and assume that $\delta>0$ (otherwise $\delta-\epsilon<0$ and the proof is trivial).
  We say that a vertex $v\in V$ is \emph{valid} w.r.t. $f$ if $f(v,\cdot)\in\RS[n,k]$.
  Let $T:=\{v\in V\mid f(v,\cdot)\notin\RS[n,k]\}$ be the vertices of $\Gamma$ that are not valid w.r.t. $f$, and $T':=\bigcup_{\icoupes=0}^{\nombrecoupes-1}\varphi_\icoupes^{-1}(V_\icoupes\cap T)\subseteq V_0$ be the set of vertices $v_0\in V_0$ for which $\pliage{f}{\randomness}(v_0,\cdot)$ is built from at least a vertex $\varphi_\icoupes(v_0)\in V_\icoupes$ for $\icoupes\in\{0,...,\nombrecoupes\}$ that is not valid w.r.t. $f$.

  We now prove \cref{eq:inequality-Tprime}.
  \begin{equation}
    \label{eq:inequality-Tprime}
    |T'|_{\tilde{w}}\geq\frac{|T|_w}{\nombrecoupes}=\frac{\delta|V|_w}{\nombrecoupes}=\delta|V_0|_{\tilde{w}}\text.
  \end{equation}
  By using \cref{eq:symmetry-cuts,eq:symmetry-weight}, we have
  \begin{equation}
    |V|_w
    =\sum_{v\in V}w(v)
    =\sum_{\icoupes=0}^{\nombrecoupes-1}\sum_{v\in V_\icoupes}\frac{w(v)}{\superposition{v}}
    =\sum_{\icoupes=0}^{\nombrecoupes-1}\sum_{v\in V_0}\frac{w(\varphi_\icoupes(v))}{\superposition{\varphi_\icoupes(v)}}
    =\nombrecoupes|V_0|_{\tilde{w}}\text.\label{eq:inequality-V-V0}
  \end{equation}
  Furthermore, since for any $\icoupes\in\{0,..,\nombrecoupes-1\},\varphi_\icoupes^{-1}(V_\icoupes\cap T)\subseteq T'$, we have by similar arguments that
  \begin{equation}
    |T|_w
    =\sum_{\icoupes=0}^{\nombrecoupes-1}\sum_{v\in V_\icoupes\cap T}\tilde{w}(v)
    =\sum_{\icoupes=0}^{\nombrecoupes-1}\sum_{v\in \varphi_\icoupes^{-1}(V_\icoupes\cap T)}\tilde{w}(\varphi_\icoupes(v))
    \leq\sum_{\icoupes=0}^{\nombrecoupes-1}\sum_{v\in T'}\tilde{w}(v)
    =\nombrecoupes|T'|_{\tilde{w}}\text.\label{eq:inequality-T-Tprime}
  \end{equation}
  By definition of the distance, $|T|_w=\delta|V|_w$.
  Hence by \cref{eq:inequality-V-V0,eq:inequality-T-Tprime}, we obtain \cref{eq:inequality-Tprime} as expected.

  For $\randomness\in\bbF$, denote
  \[
    S_\randomness:=\{v_0\in T'\mid \pliage{f}{\randomness}(v_0,\cdot)\in\RS[n,k]\}
  \]
  the vertices of $\Gamma'$ whose fold was built from non-valid vertices w.r.t. $f$, but are made valid w.r.t. $\pliage{f}{\randomness}$ because of $\randomness$.
  Then
  \begin{align*}
    &\Pr(\Delta_{\tilde{w}}(\pliage{f}{\randomness},\RS[n,k])<\delta-\epsilon)\\
    &=\Pr(|\{v_0\in V_0\mid \pliage{f}{\randomness}(v_0,\cdot)\notin\RS[n,k]\}|_{\tilde{w}}<(\delta-\epsilon)|V_0|_{\tilde{w}})\\
    &=\Pr(|\{v_0\in T'\mid \pliage{f}{\randomness}(v_0,\cdot)\notin\RS[n,k]\}|_{\tilde{w}}<(\delta-\epsilon)|V_0|_{\tilde{w}})\\
    &=\Pr(|S_\randomness|_{\tilde{w}}>|T'|_{\tilde{w}}-(\delta-\epsilon)|V_0|_{\tilde{w}})\\
    &\leq\Pr(|S_\randomness|_{\tilde{w}}>\delta|V_0|_{\tilde{w}}-(\delta-\epsilon)|V_0|_{\tilde{w}})\tag{by \cref{eq:inequality-Tprime}}\\
    &=\Pr(|S_\randomness|_{\tilde{w}}>\epsilon|V_0|_{\tilde{w}})\text,
  \end{align*}
  i.e. with $A:=\{\randomness\in\bbF\mid |S_\randomness|_{\tilde{w}}>\epsilon|V_0|_{\tilde{w}}\}$,
  \begin{equation}
    \label{eq:bound-with-A}
    \underset{\randomness\in\bbF}{\Pr}(\Delta_{\tilde{w}}(\pliage{f}{\randomness},C')<\delta-\epsilon)\leq\frac{|A|}{|\bbF|}\text.
  \end{equation}

  \noindent
  We now provide a bound on $|A|$.

  Let $v_0\in T'$ and denote $A_{v_0}:=\{\randomness\in\bbF\mid \pliage{f}{\randomness}(v_0,\cdot)\in\RS[n,k]\}$.
  For $\icoupes=\{0,...,\nombrecoupes-1\}$, denote $\sum_{\iaretes=0}^{d_\icoupes}a_{\icoupes,\iaretes}X^\iaretes$ the least degree polynomial interpolating $\coupe{f}{V_\icoupes}$.
  Let $d:=\max_\icoupes d_\icoupes$.
  Since $v_0\in T'$, there exists $\icoupes$ such that $\coupe{f}{V_\icoupes}(\varphi_\icoupes(v_0),\cdot)\notin\RS[n,k]$ and hence $d\geq d_\icoupes\geq k$.

  Moreover, $\pliage{f}{\randomness}(v_0,\cdot)\in\RS[n,k]$ only if the polynomial interpolating $\pliage{f}{\randomness}(v_0,\cdot)$ has degree $<d$.
  By construction, this polynomial is
  \[
    \sum_{\iaretes=0}^d\left(\sum_{\icoupes=0}^{\nombrecoupes-1}a_{\icoupes,\iaretes}\randomness^\icoupes\right)X^\iaretes\text,
  \]
  hence it has degree $<d$ only if $\randomness$ is a root of the polynomial $\sum_{\icoupes=0}^{\nombrecoupes-1}a_{\icoupes,d}Y^\icoupes$, that has degree $\leq \nombrecoupes-1$.
  Therefore $|A_{v_0}|\leq \nombrecoupes-1$.

  The bound on $|A|$ is the following.
  On the one hand, by definition of $A$,
  \[
    \sum_{\randomness\in A}\sum_{v_0\in T'}\mathds{1}_{A_{v_0}}(\randomness)\tilde{w}(v_0)=\sum_{\randomness\in A}\sum_{v_0\in S_\randomness}\tilde{w}(v_0)=\sum_{\randomness\in A}|S_\randomness|_{\tilde{w}}>\epsilon|V_0|_{\tilde{w}}|A|\text,
  \]
  and on the other hand,
  \[
    \sum_{v_0\in T'}\sum_{\randomness\in A}\mathds{1}_{A_{v_0}}(\randomness)\tilde{w}(v_0)=\sum_{v_0\in T'}\tilde{w}(v_0)|A_{v_0}|\leq(\nombrecoupes-1)|T'|_{\tilde{w}}\leq(\nombrecoupes-1)|V_0|_{\tilde{w}}\text.
  \]
  Thus $|A|\leq\frac{\nombrecoupes-1}{\epsilon}$ and by \cref{eq:bound-with-A} we obtain the result.\QED
\end{proof}


\subsection{Flowering protocol}

Let $\Gamma_0=(V_0,E_0),...,\Gamma_\rounds=(V_\rounds,E_\rounds)$ be a blossoming graph sequence.
For $\irounds\in[\rounds]$, let $(V_{\irounds,0},...,V_{\irounds,\nombrecoupes-1})$ be a flowering cut collection from $\Gamma_{\irounds-1}$ to $\Gamma_\irounds$.

\begin{definition}[Flowering protocol]
  \label{definition:flowering-protocol}
  The flowering protocol is composed of two phases, the commit phase which is interactive, and the query phase, where the Verifier opens the commits.
  There are two complexity parameters that one can fix: the number $\repetitions$ of repetitions of the query phase, and the proportion $\proportioncheck\in[0,1]$ of edges of $[n]$ checked at each step.
  On input $f_0\in\mots{\Gamma_0}$, the protocol runs as follows:

  \textsc{Commit phase:} for $\irounds$ from $0$ to $\rounds-1$, the Verifier picks and send $\randomness_{\irounds}\overset{\$}{\leftarrow}\bbF$ to the Prover, and the Prover provides to the Verifier an oracle access to a word $f_{\irounds+1}\in\mots{\Gamma_{\irounds+1}}$.

  \textsc{Query phase:} for $\irepetitions\in[\repetitions]$, the Verifier picks $v_{0,\irepetitions}\overset{\$}{\leftarrow}V_0$ and a random set $\sousensemblearetes_\irepetitions\subseteq[n]$ of size $\proportioncheck n$.
  For $\irounds\in[\rounds]$, the Verifier picks randomly $\icoupes_{\irounds,\irepetitions}\overset{\$}{\leftarrow}\{\icoupes\mid v_{\irounds,\irepetitions}\in V_{\irounds,\icoupes}\}$, computes $v_{\irounds,\irepetitions}:=\varphi_{\irounds-1,\icoupes_{\irounds,\irepetitions}}^{-1}(v_{\irounds-1,\irepetitions})$, and checks that
  \begin{equation*}
    \label{eq:check-fold}
    \forall\iaretes\in \sousensemblearetes_\irepetitions, \pliage{f_{\irounds-1}}{\randomness_{\irounds-1}}(v_{\irounds,\irepetitions},\iaretes)=f_\irounds(v_{\irounds,\irepetitions},\iaretes)
  \end{equation*}
  by making $\proportioncheck n\nombrecoupes$ queries if $\irounds=1$, or $\proportioncheck n(\nombrecoupes-1)$ queries if $\irounds\geq 2$, to $f_{\irounds-1}$ and $\proportioncheck n$ queries to $f_\irounds$.
  Finally with $v_\rounds$ the flower of $\Gamma_\rounds$, the Verifier checks that
  \begin{equation*}
    \label{eq:check-final}
    f_\rounds(v_\rounds,\cdot)\in\RS[n,k]\text.
  \end{equation*}
  The Verifier accepts if and only if all checks pass.
\end{definition}

\noindent
The query phase is adapted from~\cite{DMR25} to the cases where the cuts are not disjoint.
When there are $\nombrecoupes=2$ disjoint cuts, the Verifier does not have a choice for the next vertex $v_{\irounds+1,\irepetitions}$: either $v_{\irounds,\irepetitions}\in V_{\irounds,0}$ and $v_{\irounds+1,\irepetitions}:=v_{\irounds,\irepetitions}$, either $v_{\irounds,\irepetitions}\notin V_{\irounds,0}$ and $v_{\irounds+1,\irepetitions}:=\varphi_{\irounds,1}^{-1}(v_{\irounds,\irepetitions})$.
Here, $v_{\irounds,\irepetitions}$ belongs to $\superposition{v_\irounds}$ different $V_{\irounds,\icoupes}$'s, hence the Verifier has to pick one of these to pick $v_{\irounds+1,\irepetitions}$.

Since picking uniformly a subset $\sousensemblearetes_\irepetitions\subseteq[n]$ of size $\proportioncheck n$ is hard for $\proportioncheck\approx 1/2$, one can focus on the two extreme cases $\proportioncheck=1/n$, to focus on the Verifier complexity, or $\proportioncheck=1$, to focus on the soundness, as we do in \Cref{tab:comparison-complexity,tab:comparison-soundness}.

We give the complexity properties of the protocol in \Cref{theorem:complexities}, and the soundness property in \Cref{theorem:soundness}.
The protocol easily has perfect completeness: the Prover just has to send $f_\irounds=\pliage{f_{\irounds-1}}{\randomness_{\irounds-1}}$ at each step $\irounds\in[\rounds]$.

\begin{theorem}[Complexity of the flowering protocol {\cite[Theorem 1]{DMR25}}]
  \label{theorem:complexities}
  The flowering protocol defined in \Cref{definition:flowering-protocol} has the following complexities
  \begin{itemize}
  \item Prover complexity: $(2\nombrecoupes-1)\sum_{\irounds=1}^\rounds|\aretes{\Gamma_\irounds}|<(2\nombrecoupes-1)\rounds|\aretes{\Gamma_0}|$ field operations
  \item Verifier complexity: $\rounds\repetitions\proportioncheck n(\nombrecoupes+2)+n\log n$ field operations
  \item Query complexity: $(\nombrecoupes\rounds+\nombrecoupes-1)\proportioncheck n\repetitions+n$
  \item Round complexity: $\rounds$
  \item Randomness complexity: $\rounds$ field elements and cuts, $\repetitions$ vertices and subsets of $[n]$
  \item Proof length: $\sum_{\irounds=1}^\rounds|\aretes{\Gamma_\irounds}|<\rounds|\aretes{\Gamma_0}|$ field elements
  \end{itemize}
\end{theorem}

\noindent
To obtain the soundness of the flowering protocol, stated in vertex distance in \Cref{proposition:query-soundness}, we adapt the proof of \cite[Proposition 4]{DMR25} using \Cref{proposition:commit-soundness} instead of \cite[Proposition 3]{DMR25} and the weighted vertex distances instead of the uniform vertex distance.
The proof is done in \Cref{sec:proof-query-soundness}.

\begin{proposition}[Query soundness]
  \label{proposition:query-soundness}
  Let $\epsilon>0$, $k<n$ and $f_0\in\mots{\Gamma_0}$.
  After running the Flowering protocol with \repetitions{} repetitions of the query phase by checking $\proportioncheck n$ edges, the probability over its internal randomness that the Verifier accepts is at most
  \begin{equation*}
    \label{eq:query-soundness-vertex-distance}
    \frac{\rounds(\nombrecoupes-1)}{\epsilon|\bbF|}+\left(1-\proportioncheck\cdot\big(\Delta_V(f_0,\code{\Gamma_0}{k})-\rounds\epsilon\big)\right)^\repetitions\text.
  \end{equation*}
\end{proposition}

\noindent
If the successive cut collections have a different order $\nombrecoupes_\irounds$, then the left term of the soundness becomes $\sum_{\irounds=1}^\rounds\frac{\nombrecoupes_\irounds-1}{\epsilon|\bbF|}$.

Combining \Cref{proposition:inequality-distances,proposition:query-soundness}, we obtain the soundness of the flowering protocol, stated in \Cref{theorem:soundness}, stated here in Hamming distance.

\begin{theorem}[Soundness of the flowering protocol]
  \label{theorem:soundness}
  Let $k<n$ and $f_0\in\mots{\Gamma_0}$.
  After running the Flowering protocol with \repetitions{} repetitions of the query phase by checking $\proportioncheck n$ edges, the probability over its internal randomness that the Verifier accepts is at most
  \begin{equation*}
    \label{eq:query-soundness}
    \min_{\epsilon>0}\left(\frac{\rounds(\nombrecoupes-1)}{\epsilon|\bbF|}+\left(1-\proportioncheck\cdot\big(\Delta_H(f_0,\code{\Gamma_0}{k})-\rounds\epsilon\big)\right)^\repetitions\right)\text.
  \end{equation*}
\end{theorem}

\noindent
Here again, if the successive cut collections have a different order $\nombrecoupes_\irounds$, the left term of the soundness becomes $\sum_{\irounds=1}^\rounds\frac{\nombrecoupes_\irounds-1}{\epsilon|\bbF|}$.




\section{Cutting Cayley graphs}

We describe how to apply the flowering protocol using any Cayley graph~\cite{Cay78}.

\begin{definition}[Cayley graph]
  Let $G$ be a finite multiplicative group and $S=\{s_1, ...,s_n\}$ a symmetric generating set.
  We define the corresponding Cayley graph $\Cay(G,S)$ by $V=G$,
  and $E : (g,\iaretes)\mapsto g\cdot s_\iaretes$.
\end{definition}

\noindent
We also allow ourselves to write the edges as a function $E:G\times S\to G$ where $E(g,s):=gs$.

\Cref{fig:cayley-graph} illustrates the construction of a Cayley graph.
\begin{figure}[h!]
  \centering
  \scalebox{0.7}{
    \begin{tikzpicture}[every node/.style={inner sep=1.5pt}, scale=0.5]

      \node (a1) at (0*360/3+90:1.5) {$Id$};
      \node (a2) at (1*360/3+90:1.5) {$(132)$};
      \node (a3) at (2*360/3+90:1.5) {$(123)$};

      \begin{scope}[shift={(90:6)}]
        \node (b1) at (0*360/3-90:1.5) {$(12)(34)$};
        \node (b2) at (1*360/3-90:1.5) {$(143)$};
        \node (b3) at (2*360/3-90:1.5) {$(243)$};
      \end{scope}

      \begin{scope}[shift={(90+360/3:6)}]
        \node (c1) at (0*360/3+30:1.5) {$(234)$};
        \node (c2) at (1*360/3+30:1.5) {$(142)$};
        \node (c3) at (2*360/3+30:1.5) {$(13)(24)$};
      \end{scope}

      \begin{scope}[shift={(90-360/3:6)}]
        \node (d1) at (0*360/3+150:1.5) {$(134)$};
        \node (d2) at (1*360/3+150:1.5) {$(14)(23)$};
        \node (d3) at (2*360/3+150:1.5) {$(124)$};
      \end{scope}

      \draw[areteA] (a1) to (a2);
      \draw[areteA] (a2) to (a3);
      \draw[areteA] (a3) to (a1);

      \draw[areteA] (b1) to (b2);
      \draw[areteA] (b2) to (b3);
      \draw[areteA] (b3) to (b1);

      \draw[areteA] (c1) to (c2);
      \draw[areteA] (c2) to (c3);
      \draw[areteA] (c3) to (c1);

      \draw[areteA] (d1) to (d2);
      \draw[areteA] (d2) to (d3);
      \draw[areteA] (d3) to (d1);

      \draw[areteA] (a1) to (b1);
      \draw[areteA] (a2) to (c1);
      \draw[areteA] (a3) to (d1);

      \draw[areteA] (b3) to (c2);
      \draw[areteA] (c3) to (d2);
      \draw[areteA] (d3) to (b2);

    \end{tikzpicture}
  }
  \caption{Cayley graph on the alternating group $G=\calA_4$ and $S=\{(123), (132), (12)(34)\}$.}
  \label{fig:cayley-graph}
\end{figure}
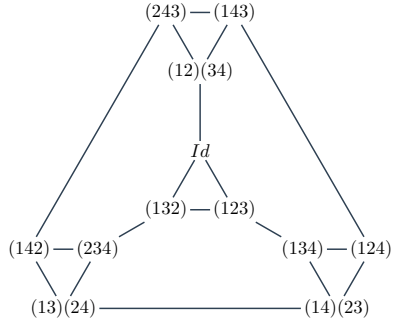
Recall that in a graph $\Gamma=(V,E)$, a \emph{path} of length $\ell$ is a sequence of vertices $v_0,...,v_\ell$ such that for all $i\in[\ell]$, $v_{i-1}$ and $v_i$ are neighbors.
The \emph{distance} between two vertices $v,v'\in V$ is the length $\ell$ of the shortest path $v_0,...,v_\ell$ such that $v_0=v$ and $v_\ell=v'$.
The \emph{diameter} of a graph, denoted $\diam(\Gamma)$, is the maximum distance between two vertices of $V$.

\begin{proposition}
  \label{proposition:cut-friendly-writing}
  Let $G$ be a finite group.
  Fix $S_1 = \{s_0, ..., s_{\gne-1}\}$ a generative subset of $G$ and consider $S = S_1 \cup S_{-1}$, where $S_{-1} = \{s_0^{-1}, ... ,s_{\gne-1}^{-1}\}$, the induced symmetric generating set.
  Define $\boucle{s}_\icoupes:=s_{\icoupes\mod\gne}$.
  Let $\Gamma:=\Cay(G,S)$.
  Then with $\rounds:=\gne\,\diam(\Gamma)$, we have
  \begin{equation}
    \label{eq:final-writing-G}
    G=\left\{\boucle{s}_0^{\alpha_0}\cdots\boucle{s}_{\rounds-1}^{\alpha_{\rounds-1}}\mid \alpha_0,...,\alpha_{\rounds-1}\in\{-1,0,1\}\right\}\text.
  \end{equation}
\end{proposition}

\begin{proof}
  Since $S$ generates $G$, for any $g\in G$, there exists $\ell\in\bbN$ and a sequence $\path=\left(s_{i_0}^{\beta_0},...,s_{i_{\ell-1}}^{\beta_{\ell-1}}\right)\in S^\ell$ where $\beta_j\in\{-1,1\}$ such that $g=s_{i_0}^{\beta_0}\cdots s_{i_{\ell-1}}^{\beta_{\ell-1}}$.
  We can see $\path$ as a path from $1_G$ to $g$ in $\Gamma$. By definition of the diameter, there exists such a path of length $\ell\leq d:=\diam(\Gamma)$.
  Hence, any $g\in G$ can be written with a sequence of $d$ elements $w:=(s_{i_0}^{\gamma_0},...,s_{i_{d-1}}^{\gamma_{d-1}})\in S^d$ with $\gamma_j\in\{-1,0,1\}$, such that $g=s_{i_0}^{\gamma_0}\cdots s_{i_{d-1}}^{\gamma_{d-1}}$.
  We obtain the following description of $G$:
  \begin{equation}
    \label{eq:first-writing-G}
    G=\left\{s_{i_0}^{\gamma_0}\cdots s_{i_{d-1}}^{\gamma_{d-1}}\mid
      \begin{cases}
        i_0,...,i_{d-1}\in \{0,...,\gne-1\}\\
        \gamma_0,...,\gamma_{d-1}\in\{-1,0,1\}
      \end{cases}
    \right\}\text.
  \end{equation}

  \noindent
  Given $g\in G$ associated with a sequence $w=(s_{i_0}^{\gamma_0},...,s_{i_{d-1}}^{\gamma_{d-1}})$, let us construct a sequence $(\boucle{s}_0^{\alpha_0},...,\boucle{s}_{\rounds-1}^{\alpha_{\rounds-1}})$ such that $g=\boucle{s}_0^{\alpha_0}\cdots\boucle{s}_{R-1}^{\alpha_{\rounds-1}}$.
  For $s_j^{\gamma}$ in $w$, denote for $i\in\{0,...,\gne - 1\}$
  \[
    \alpha_{j,i}:=
    \begin{cases}
      \gamma&\text{if }i=j,\\
      0&\text{otherwise.}
    \end{cases}
  \]
  Then $s_{0}^{\alpha_{j,0}} ... s_{\gne - 1}^{\alpha_{j,\gne - 1}} = s_{j}^{\alpha_{j,j}} = s_j^{\gamma}$.
  Therefore with $\tilde{\alpha}_{j\gne+i}:=\alpha_{j,i}$, $g$ can be written
  $$g=\prod_{j = 0} ^{d-1} \prod_{i=0}^{\gne-1} s_{i}^{\alpha_{j,i}} = \prod_{\irounds=0}^{\rounds-1} \boucle{s}_{\irounds}^{\tilde{\alpha}_{\irounds}}.$$
  Thus, by \cref{eq:first-writing-G} we obtain \cref{eq:final-writing-G}.
  \QED
\end{proof}

\noindent
Thanks to the writing of $G$ as in \cref{eq:final-writing-G} using \Cref{proposition:cut-friendly-writing}, we can define a blossoming sequence of length $\rounds$ by taking shorter and shorter paths.
For $\icoupes\in\{-1,0,1\}$ and $\irounds\in[\rounds]$, define
\begin{equation*}
  \label{eq:def-cut}
  V_{\irounds,\icoupes}:=\left\{\boucle{s}_{\irounds-1}^\icoupes\cdot \boucle{s}_{\irounds}^{\alpha_\irounds}\cdots \boucle{s}_{\rounds-1}^{\alpha_{\rounds-1}}\mid \alpha_\irounds,...,\alpha_{\rounds-1}\in\{-1,0,1\}\right\}\text.
\end{equation*}
We have that $V_{\rounds,0}=\{1_G\}$ is a singleton and we prove in \Cref{proposition:cut-equivalence} that for $\irounds\in[\rounds]$, the cuts $\coupe{\Gamma_{\irounds-1}}{V_{\irounds,\icoupes}}$ for $\icoupes\in\{-1,0,1\}$ are isomorphic, and in \Cref{proposition:cayley-cut-are-symmetric} that they form a flowering cut collection, therefore the graphs form a blossoming graph sequence $\Gamma_0,...,\Gamma_\rounds$.
For $\irounds\in[\rounds]$, we denote $\Gamma_\irounds=(V_{\irounds,0},E_\irounds):=\coupe{\Gamma_{\irounds-1}}{V_{\irounds,0}}$, and for $\icoupes\in\{-1,0,1\}$, we denote $\Gamma_{\irounds,\icoupes}:=\coupe{\Gamma_{\irounds-1}}{V_{\irounds,\icoupes}}$.

\begin{proposition}
  \label{proposition:cut-equivalence}
  Let $\irounds\in[\rounds]$ and $\icoupes\in\{-1,0,1\}$. With the notations above, $\varphi_{\irounds,\icoupes}:g\mapsto\boucle{s}_{\irounds-1}^{\icoupes}\cdot g$ is an isomorphism from $\Gamma_{\irounds,0}$ to $\Gamma_{\irounds,\icoupes}$.
\end{proposition}

\begin{proof}
  Let $g\in V_{\irounds,0}$ and $s\in S$.
  Then $E_{\irounds}(g,s)=gs$ if $gs\in V_{\irounds,0}$ and $g$ otherwise,
  and $E_{\irounds}(\varphi_{\irounds,\icoupes}(g),s)=\boucle{s}_{\irounds}^{\icoupes}\cdot gs$ if $\boucle{s}_{\irounds}^{\icoupes}\cdot gs\in V_{i,0}$ and $\boucle{s}_{\irounds}^{\icoupes}\cdot g$ otherwise.
  Thus we only have to prove that $gs\in V_{\irounds,0}$ iff $\boucle{s}_{\irounds}^{\icoupes}\cdot gs\in V_{\irounds,\icoupes}$.
  We have that $gs\in V_{\irounds,0}$ iff there exists $\alpha_{\irounds+1}, ... \alpha_{\rounds}$ such that

  \[
    gs=\boucle{s}_{\irounds}^{0}\prod_{k=\irounds+1}^{\rounds}\boucle{s}_{k}^{\alpha_k}\text,
  \]
  i.e.
  \[
    \boucle{s}_{\irounds}^{\icoupes}\cdot gs = \boucle{s}_{\irounds}^{\icoupes}\prod_{k=\irounds+1}^{\rounds}\boucle{s}_{k}^{\alpha_k}\text.
  \]
  The existence of such a decomposition is equivalent to $\boucle{s}_{\irounds}^{\icoupes}\cdot gs\in V_{\irounds,\icoupes}$.
  \QED
\end{proof}

\begin{proposition}
  \label{proposition:cayley-cut-are-symmetric}
  For $\irounds\in[\rounds]$, $\icoupes\in\{-1,0,1\}$ and $v\in V_{\irounds,0}$, we have $\superposition{v}=\superposition{\varphi_{\irounds,\icoupes}(v)}$.
\end{proposition}

\begin{proof}
  Let $g\in V_{\irounds,0}$ and $\icoupes_0\in\{-1,0,1\}$.
  For $\icoupes\in\{-1,0,1\}$, let $\icoupes_1:=\icoupes+\icoupes_0\mod2$.
  Since $g\in V_{\irounds,\icoupes}$ iff $g=\boucle{s}_{\irounds-1}^\icoupes\cdot\boucle{s}_\irounds^{\alpha_\irounds}\cdots\boucle{s}_{\rounds-1}^{\alpha_{\rounds-1}}$ for some $\alpha_\irounds,...,\alpha_{\rounds-1}\in\{-1,0,1\}$, we have that $g\in V_{\irounds,\icoupes}$ iff $\varphi_{\irounds,\icoupes_0}(g)=\boucle{s}_{\irounds-1}^{\icoupes_0}\cdot g\in V_{\irounds,\icoupes_1}$.
  Hence $\superposition{v}=\superposition{\varphi_{\irounds,\icoupes_0}(v)}$.
  \QED
\end{proof}

\noindent
A blossoming graph sequence of the Cayley graph $\Cay(\calA_4,S)$ defined in \Cref{fig:cayley-graph} is illustrated in \Cref{fig:blossoming-sequence}. 
\Cref{proposition:cut-friendly-writing} was shown using $\{-1,0,1\}$ as decomposition set, but the result also holds for any superset of $\{-1,0,1\}$, or for different sets for each coefficient $s_\irounds$.
For instance if $s_\irounds$ has order $\nombrecoupes_\irounds$ in $G$, then one could use the decomposition set $\{0,1,...,\nombrecoupes_\irounds-1\}$.


\section{Application to expander families of Cayley graphs}
\label{sec:application}

\subsection{Expansion in graphs}
In this section we recall spectral properties of expander graphs, and we establish bounds on parameters relevant to our setting.

\begin{definition}[Relative spectral gap]
    Let $\Gamma=(V,E)$ be a $n$-regular graph and $\left\{\lambda_1,...,\lambda_{|V|}\right\}$
    the eigenvalues of its adjency matrix, ordered by decreasing absolute value:
    $|\lambda_1|\geq ... \geq \left|\lambda_{|V|}\right|$.
    Then $\lambbar = \frac{\lambda_2}{\lambda_1} = \frac{\lambda_2}{n}$ is the \emph{relative spectral gap} of $\Gamma$.

\end{definition}
 A smaller $\lambbar$ corresponds to stronger expansion properties of the graph.
 For any $\lambgne \geq \lambbar$, we say that $\Gamma(V,E)$ is a $\lambgne$-expander graph.
 A family of graphs $(\Gamma_i=(V_i,E_i))_{i\in\mathbb{N}}$ is said to have \textbf{constant expansion}
 if there exist a $\lambgne$ such that for all $i\in\mathbb{N}$, $\Gamma_i$ is a $\lambgne$-expander.

\begin{definition}[Ramanujan Graph]
  \label{def:ramanujan}
  A $n$-regular graph with relative spectral expansion $\lambbar$ is a Ramanujan graph if and only if $\lambbar \leq 2\frac{ \sqrt{n-1}}{n}$.
\end{definition}
Ramanujan graphs are those whose spectral expansion approaches the Alon-Boppana lower bound established in \cite{Nil91}.
As such, they achieve the best possible spectral expansion.

We give the two properties that interest us when building codes on expander graphs in \Cref{proposition:bound-on-the-distance,proposition:bound-on-the-diameter}.

\begin{proposition}[Bound on minimum distance]
  \label{proposition:bound-on-the-distance}
  Let a $\lambgne$-expander n-regular graph $\Gamma(V,E)$ be given.
  Denote $\delta:=1-\frac{k-1}{n}$ the relative minimum distance of $\RS[n,k]$. Then:
  $$\Delta(\code{\Gamma}{k}) \geq \delta ( \delta - \lambgne).$$
\end{proposition}
\begin{proof}
  Let $\Gamma=(V,E)$ be a $\lambgne$-expander $n$-regular graph and let $c\in\code{\Gamma}{k}$ be a nonzero word.
  Since $c\neq 0$, there exists $v_0\in V$ such that at least one of its incident edges has nonzero value.
  For clarity, we will say that such a vertex is nonzero.
  Let $T\subset V$ be the smallest set such that $\forall v\in T$, $v\neq 0$.

  If $c(v,\cdot)\in \text{RS}[n,k] \neq 0$, then at least $n-k+1$ of its components are nonzero.
  Therefore, $T$'s average arity is greater or equal to $n-k+1=\delta n$.
  From~\cite[Lemma 2.1]{DELLM21}, we get that $|T| \geq (\delta - \lambgne)N$.
  \QED
\end{proof}

\begin{proposition}[Bound on the diameter]
  \label{proposition:bound-on-the-diameter}
  Let $\Gamma=(V,E)$ be a $n$-regular $\lambgne$-expander connected graph.
  Then we have:
  \[
    \diam(\Gamma) \leq 2\frac{\log\left(\frac{|V|}{2}\right)}{\log\left(\frac{3-\lambgne}{2}\right)} + 3
  \]
\end{proposition}
\begin{proof}
  A similar bound involving combinatorial expansion is established~\cite[Proposition 3.1.5]{K19}.
  We infer this spectral version of the bound using Cheeger inequality.
\end{proof}
\noindent
It naturally gives a bound on the round complexity of the protocol:
$$\rounds \leq 2\gne\frac{\log\left(\frac{|V|}{2}\right)}{\log\left(\frac{3-\lambgne}{2}\right)} + 3\gne\text.$$


\subsection{Case study: Ramanujan Cayley graphs family}
\label{sec:case-study}

In this part, we instanciate the protocol on codes built on a family of Cayley graphs with good expansion properties.
The code and protocol parameters are explicitely computed for such graphs.

Several families of Cayley graphs have been proven to be expander families \cite{Lub12}.
Among them, an explicit construction of Ramanujan expander graphs family was established independently by Margulis \cite{Mar88} and by Lubotzky, Phillips and Sanark \cite{LPS88}.
A concise description of this construction can be found in \cite[Section 11]{Mur20}.

\begin{proposition}[\cite{Mur20}]
  Fix $p$ prime such that $p \equiv 1 [4]$. For any $q\neq p$ prime such that $q \equiv 1 [4]$, denote:
  \[
    G_{p,q} =
    \begin{cases}
      \PSL_2(\bbF_q) & \text{if $p$ is a quadratic residue mod $q$},  \\
      \PGL_2(\bbF_q) & \text{otherwise.}\\
    \end{cases}
  \]
  \noindent
  A symmetric generating set $S_{p,q}$ of size $p+1$ can be explicitely built such that $\Cay(G_{p,q}, S_{p,q})$ is a Ramanujan graph.
\end{proposition}

\noindent
As a consequence, $p$ being a fixed prime, one can construct a family of $(p+1)$-regular Ramanujan Cayley graphs whose size can grow arbitrarily with $q$.
By definition of a Ramanujan graph, such a family has an upper bound $\lambgne = 1/\sqrt{p}$ on its relative spectral expansion that only depends on $p$.

Fix $p$ prime and select $q$ such that $G_{p,q} = \PGL_2(\bbF_q)$.
Denote $\Gamma_{p,q}(V,E)$ the associated Ramanujan graph. It has $|\PGL_2(\bbF_q)| = (q-1)q(q+1)$ vertices.
Consider the graph code $\code{\Gamma_{p,q}}{k}$ over the basecode $\text{RS}[p+1,k]$, $(p+1)/2<k<p+1$.
The code has length
\[
  N = (p+1)(q-1)q(q+1)/2\text.
\]

\noindent
The code dimension is
\[
  \dim \code{\Gamma_{p,q}}{k}\geq \left(k - \frac{p+1}{2}\right)(q-1)q(q+1)\text.
\]

\noindent
From \Cref{proposition:bound-on-the-distance} we get $\code{\Gamma_{p,q}}{k}$ minimum relative distance
\[
  \Delta\left(\code{\Gamma_{p,q}}{k}\right) \geq \delta ( \delta - \lambgne)\text,
\]
where $\delta = \frac{p-k+2}{p+1}$ is the minimum relative distance of $\text{RS}[p+1,k]$.
Since $\lambgne$ only depends on $p$, which does not depend on $q$, the minimum relative distance remains constant as the code length grows.
The diameter of the graph is also bounded according to \Cref{proposition:bound-on-the-diameter}, which yields a bound on the number of round:
\[
  \rounds \leq \constant(p) \log(N)\text,
\]
where $\constant(p)$ is constant with respect to $q$.


\section{Conclusion}

This paper generalized the IOPP for codes on $(2,n)$-regular Tanner graphs introduced in \cite{DMR25} that was only applied on codes with $o(1)$ minimum distance to any $(2,n)$-regular Tanner graph built as a Cayley graph.
This generalization allows to use Cayley graphs with good expansion, yielding codes with constant rate and minimum distance.
The great soundness parameter from~\cite{DMR25} compared to the FRI protocol~\cite{BCIKS23} is thus preserved, even if the complexity becomes slighly higher, while now allowing the use of codes with constant rate and constant minimum distance.


\section*{Acknowledgments}
The authors deeply thank Tanguy Medevielle and Élina Roussel, with whom this work was initiated as the continuation of \cite{DMR25} that first defines the flowering protocol.
The authors also thank Daniel Augot for his proofreadings and guidance and Pierre Loisel for his help regarding group theory and Cayley graphs.

\bibliographystyle{alpha}
\bibliography{biblio}

\appendix

\section{Proof of the query soundness}
\label{sec:proof-query-soundness}

We first prove \Cref{lemma:disjoint-events,lemma:weighted-probabilities} in order to prove \Cref{proposition:query-soundness}.

\begin{lemma}[{\cite[Lemma 1]{DMR25}}]
  \label{lemma:disjoint-events}
  Using the notations of \Cref{definition:flowering-protocol}, let $\irepetitions\in[\repetitions]$ and $\icoupes_{1,\irepetitions},...,\icoupes_{\rounds,\irepetitions}$ be fixed.
  There exists $(f'_{0,\irepetitions},...,f'_{\rounds,\irepetitions})\in\prod_{\irounds=0}^\rounds\mots{\Gamma_\irounds}$ such that $f'_{\rounds,\irepetitions}=f_\rounds$, and by denoting
  $N_{\irounds,\irepetitions}$ the event ``$\pliage{f_{\irounds-1}}{\randomness_{\irounds-1}}(v_{\irounds,\irepetitions},\cdot)\neq f_\irounds(v_{\irounds,\irepetitions},\cdot)$''
  and $N'_{\irounds,\irepetitions}$ the event ``$\pliage{f'_{\irounds-1,\irepetitions}}{\randomness_{\irounds-1}}(v_{\irounds,\irepetitions},\cdot)\neq f'_{\irounds,\irepetitions}(v_{\irounds,\irepetitions},\cdot)$'',
  we have that $N'_{1,\irepetitions},...,N'_{\rounds,\irepetitions}$ are disjoints and $\bigsqcup_{\irounds=1}^\rounds N'_{\irounds,\irepetitions}\subseteq\bigcup_{\irounds=1}^\rounds N_{\irounds,\irepetitions}$.
\end{lemma}

\begin{proof}
  Define recursively for $\irounds=0,...,\rounds$ the ``honest'' function $\tilde{f}_\irounds$ by $\tilde{f}_0:=f_0$ and for $\irounds\geq 1$, $\tilde{f}_\irounds:=\pliage{f_{\irounds-1}}{\randomness_{\irounds-1}}$.
  The $f'_{\irounds,\irepetitions}$ are also defined recursively as follows.
  First, $f'_{0,\irepetitions}:=f_0$.
  Then let $v_0\in V_0$ and denote $(v_1,...,v_\rounds)$ the sequence such that for $\irounds\in[\rounds]$, $v_\irounds=\varphi_{\irounds-1,\icoupes_{\irounds,\irepetitions}}^{-1}(v_{\irounds-1})$ is the vertex to be tested after $v_{\irounds-1}$ by the Verifier.
  Denote
  \begin{equation*}
    \label{eq:iroundsmax}
    \iroundsmax(v_0):=\max(\{\irounds\in[\rounds]\mid f_\irounds(v_\irounds,\cdot)\neq\pliage{f_{\irounds-1}}{\randomness_{\irounds-1}}\}\cup\{0\})\text.
  \end{equation*}
  Let $\iaretes\in[n]$.
  For $\irounds<\iroundsmax(v_0)$, define $f'_{\irounds,\irepetitions}(v_\irounds,\iaretes):=\tilde{f}_\irounds(v_\irounds,\iaretes)$ and for $\irounds\geq\iroundsmax(v_0)$, define $f'_{\irounds,\irepetitions}(v_\irounds,\iaretes):=f_\irounds(v_\irounds,\iaretes)$.

  We now prove that the $N'_{\irounds,\irepetitions}$ are disjoint.
  Denote $\irounds_0:=\iroundsmax(v_{0,\irepetitions})$.
  If $\irounds_0=0$ then for any $\irounds\in[\rounds]$, $N'_{\irounds,\irepetitions}$ does not hold, and if $\irounds_0\geq 1$ then for $\irounds\in[\rounds]$, $N'_{\irounds,\irepetitions}$ holds iff $\irounds=\irounds_0$.
  Therefore there is at most one $\irounds\in[\rounds]$ such that $N'_{\irounds,\irepetitions}$ holds, hence they are disjoint.

  Furthermore, if one $N'_{\irounds,\irepetitions}$ holds, i.e. if $\bigsqcup_{\irounds=1}^\rounds N'_{\irounds,\irepetitions}$ hords, then $\irounds_0>0$ and hence $\bigcup_{\irounds=1}^\rounds N_{\irounds,\irepetitions}$ holds.
  Thus $\bigsqcup_{\irounds=1}^\rounds N'_{\irounds,\irepetitions}\subseteq\bigcup_{\irounds=1}^\rounds N_{\irounds,\irepetitions}$.\QED
\end{proof}

\noindent
We also use another lemma to adapt the proof of \Cref{proposition:query-soundness} to non-disjoint cuts.

\begin{lemma}
  \label{lemma:weighted-probabilities}
  Using the notations of the flowering protocol in \Cref{definition:flowering-protocol}, for $\irounds\in\{0,...,\rounds\}$ and $v\in V_\irounds$, we have that $\Pr(v_{\irounds,1}=v)=\frac{w_\irounds(v)}{|V_\irounds|_{w_{\irounds}}}$.
\end{lemma}

\begin{proof}
  We prove the result by induction on $\irounds=0,...,\rounds$.
  For $\irounds=0$, by definition for any $v\in V_0$ we have $\Pr(v_{i,0}=v)=\frac{1}{|V_0|}=\frac{w_0(v)}{|V_0|_{w_0}}$.
  Suppose that for any $v'\in V_\irounds$ we have $\Pr(v_{\irounds,1}=v')=\frac{w_\irounds(v')}{|V_\irounds|_{w_{\irounds}}}$ and let $v\in V_{\irounds+1}$.
  By the total probabilities, we have
  \begin{align*}
    \Pr(v_{\irounds+1,1}=v)
    &=\sum_{\icoupes=0}^{\nombrecoupes-1}\Pr(v_{\irounds+1,1}=v\mid v_{\irounds,1}=\varphi_\icoupes(v))\Pr(v_{\irounds,1}=\varphi_\icoupes(v))\\
    &=\sum_{\icoupes=0}^{\nombrecoupes-1}\frac{1}{\#\varphi_\icoupes(v)}\cdot\frac{w_\irounds(\varphi_\icoupes(v))}{|V_\irounds|_{w_\irounds}}\\
    &=\frac{\nombrecoupes w_{\irounds+1}(v)}{|V_\irounds|_{w_\irounds}}\text.
  \end{align*}
  Furthermore, by \cref{eq:inequality-V-V0}, we have $|V_{\irounds+1}|_{w_{\irounds+1}}=|V_\irounds|_{w_\irounds}/\nombrecoupes$, and thus the result.\QED
\end{proof}

\noindent
We can now prove \Cref{proposition:query-soundness}.

\begin{proof}[of \Cref{proposition:query-soundness}]
  If $f_\rounds\notin C_\rounds$ then the Verifier rejects with probability $1$, therefore we assume in the following that
  \begin{equation}
    \label{eq:fR-notin-CR}
    \Delta_{w_\rounds}(f_\rounds,C_\rounds)=0\text.
  \end{equation}
  Let $f'_{0,\irepetitions},...,f'_{\rounds,\irepetitions}$ be given by \Cref{lemma:disjoint-events}.
  For $(\irounds,\irepetitions)\in[\rounds]\times[\repetitions]$, denote $\rejectevent_{\irounds,\irepetitions}$ the event ``$\exists\iaretes\in \sousensemblearetes_\irepetitions, \pliage{f'_{\irounds-1,\irepetitions}}{\randomness_{\irounds-1}}(v_{\irounds,\irepetitions},\iaretes)\neq f'_{\irounds,\irepetitions}(v_{\irounds,\irepetitions},\iaretes)$'' meaning that the Verifier rejects that check.
  Then the probability of the event ``the Verifier accepts'' is at most $\Pr\left(\bigcap_{\irepetitions=1}^\repetitions\bigcap_{\irounds=1}^\rounds\overline{\rejectevent_{\irounds,\irepetitions}}\right)$.
  Denote $A$ the event ``$\forall\irounds\in[\rounds], \Delta_{w_{\irounds}}(\pliage{f_{\irounds-1}}{\randomness_{\irounds-1}},C_{\irounds})\geq\Delta_{w_{\irounds-1}}(f_{\irounds-1},C_{\irounds-1})-\epsilon$'' meaning that all $\rounds$ commits ``passed'' the commit soundness.

  By the law of total probability, we have
  \begin{equation*}
    \label{eq:total-proba}
    \Pr(\text{the Verifier accepts})\leq\Pr\left(\overline{A}\right)+\Pr\left(\bigcap_{\irepetitions=1}^\repetitions\bigcap_{\irounds=1}^\rounds\overline{\rejectevent_{\irounds,\irepetitions}}\mid A\right)\text.
  \end{equation*}
  By \Cref{proposition:commit-soundness}, we have
  \begin{equation*}
    \label{eq:sum-commit-soundness}
    \Pr\left(\overline{A}\right)\leq\frac{\rounds(\nombrecoupes-1)}{\epsilon|\bbF|}\text,
  \end{equation*}
  and by independence and identical distribution of the repetitions of the query phase, we only consider the case $\irepetitions=1$ and we have
  \begin{equation*}
    \label{eq:independence-query-phase}
    \Pr\left(\bigcap_{\irepetitions=1}^\repetitions\bigcap_{\irounds=1}^\rounds\overline{\rejectevent_{\irounds,\irepetitions}}\mid A\right)=\left(1-\Pr\left(\bigcup_{\irounds=1}^\rounds\rejectevent_{\irounds,1}\mid A\right)\right)^\repetitions\text,
  \end{equation*}
  hence we only have to lower bound $\Pr\left(\bigcup_{\irounds=1}^\rounds\rejectevent_{\irounds,1}\mid A\right)$.

  Take the notations $N_{\irounds,1}$ and $N'_{\irounds,1}$ from \Cref{lemma:disjoint-events}.
  Since $|\sousensemblearetes_1|=\proportioncheck n$, we have
  \begin{equation*}
    \label{eq:check-I1}
    \Pr\left(\bigcup_{\irounds=1}^\rounds\rejectevent_{\irounds,1}\mid A\cap\bigcup_{\irounds=1}^\rounds N_{\irounds,1}\right)\geq\proportioncheck\text,
  \end{equation*}
  and thus by \Cref{lemma:disjoint-events},
  \begin{equation*}
    \label{eq:from-B-to-N}
    \Pr\left(\bigcup_{\irounds=1}^\rounds\rejectevent_{\irounds,1}\right)
    \geq\proportioncheck\Pr\left(\bigcup_{\irounds=1}^\rounds N_{\irounds,1}\right)
    \geq\proportioncheck\Pr\left(\bigsqcup_{\irounds=1}^\rounds N'_{\irounds,1}\right)
    =\proportioncheck\sum_{\irounds=1}^\rounds\Pr(N'_{\irounds,1})\text.
  \end{equation*}
  Assuming that $A$ holds and by denoting $\delta_\irounds:=\Delta_{w_{\irounds}}(f'_{\irounds,1},C_\irounds)$, we have by \Cref{lemma:weighted-probabilities} and by triangle inequality that
  \begin{align*}
    \delta_\irounds
    &\geq\Delta_{w_\irounds}(\pliage{f'_{\irounds-1,1}}{\randomness_{\irounds-1}},C_\irounds)-\Delta_{w_\irounds}(f'_{\irounds,1},\pliage{f'_{\irounds-1,1}}{\randomness_{\irounds-1}})\\
    &\geq\delta_{\irounds-1}-\epsilon-\Pr(N'_{\irounds,1})\text.
  \end{align*}
  Thus $\Pr(N'_{\irounds,1})\geq\delta_0-\delta_\rounds-\rounds\epsilon$ and by telescoping,
  \begin{equation}
    \label{eq:inequality-Nprime-deltar}
    \sum_{\irounds=1}^\rounds\Pr(N'_{\irounds,1})\geq\delta_0-\delta_\rounds-\rounds\epsilon\text.
  \end{equation}
  Since $f'_{\rounds,1}=f_\rounds$, by \cref{eq:fR-notin-CR,eq:total-proba,eq:sum-commit-soundness,eq:independence-query-phase,eq:from-B-to-N,eq:inequality-Nprime-deltar}, we get the result.\QED
\end{proof}



\end{document}